\documentclass[letter,11pt]{article}
\usepackage[utf8]{inputenc}

\usepackage{geometry}

\usepackage{amsmath,amsthm,amssymb}

\newtheorem{theorem}{Theorem}
\newtheorem{definition}{Definition}

\newtheorem{lemma}{Lemma}

\usepackage{multirow}
\usepackage{todonotes}
\usepackage{fullpage}
\usepackage{booktabs}

\newcommand{\subsum}{\textsc{Subset Sum}}
\newcommand{\rss}{\textsc{Restricted Subset Sum}}
\newcommand{\rxc}{\textsc{Restricted Exact Cover by 3-Sets}}
\newcommand{\knapsack}{\textsc{Knapsack}}

\newcommand{\sign}{\operatorname{sign}}
\newcommand{\NP}{NP}
\newcommand{\coNP}{coNP}
\newcommand{\NPneqcoNPpoly}{\mbox{\NP{} $\not\subseteq$ \coNP/poly}}

\usepackage{xargs}

\newcommand{\problemdef}[3]{%
  \medskip
  \hspace{-1em}
    \fbox { \parbox { .95\textwidth} {\vspace{0.1cm}
    \textsc{{\hspace{-0.5em} \textsc{#1}}}
    \vspace{-0.2cm}
    \begin{description}
        \item[Input: ] #2
        \vspace{-0.2cm}
        \item[Task:]\hspace{1mm} #3
    \end{description}
    \vspace{-0.2cm}}}
    \medskip
}

\usepackage[colorlinks=true,linkcolor=blue,citecolor=red]{hyperref}
\usepackage[nameinlink]{cleveref}

\crefname{claim}{Claim}{Claims}
\Crefname{claim}{Claim}{Claims}

\newcommand{\etal}{et al.}
\newcommand{\ie}{i.e., }

\def\equationautorefname~#1\null{Equation~(#1)\null}

\crefname{equation}{ILP}{ILPs}
\Crefname{equation}{ILP}{ILPs}
\newenvironment{ILP}{\begin{equation}}{\end{equation}}

\usepackage[numbers,sort&compress]{natbib}
\author{Klaus Heeger\footnote{Department of Industrial Engineering and Management, Ben-Gurion~University~of~the~Negev, Beer-Sheva, Israel. \texttt{heeger@post.bgu.ac.il}.} \and
Danny Hermelin\footnote{Department of Industrial Engineering and Management, Ben-Gurion~University~of~the~Negev, Beer-Sheva, Israel. \texttt{hermelin@bgu.ac.il}.} \and
Matthias Mnich\footnote{Hamburg University of Technology, Institute for Algorithms and Complexity, Hamburg, Germany. \texttt{matthias.mnich@tuhh.de}. Supported by DFG grant MN 59/4-1.} \and
Dvir Shabtay\footnote{Department of Industrial Engineering and Management, Ben-Gurion~University~of~the~Negev, Beer-Sheva, Israel. \texttt{dvirs@bgu.ac.il}.}}

\title{\bfseries \LARGE No Polynomial Kernels for Knapsack\footnote{Supported by the ISF, grant No.~1070/20.}}

\date{}

\begin{document}

\maketitle

\begin{abstract}
This paper focuses on kernelization algorithms for the fundamental \knapsack\ problem. A kernelization algorithm (or kernel) is a polynomial-time reduction from a problem onto itself, where the output size is bounded by a function of some problem-specific parameter. Such algorithms provide a theoretical model for data reduction and preprocessing and are central in the area of parameterized complexity. In this way, a kernel for \knapsack\ for some parameter $k$ reduces any instance of \knapsack\ to an equivalent instance of size at most $f(k)$ in polynomial time, for some computable function $f(\cdot)$. When $f(k)=k^{O(1)}$ then we call such a reduction a polynomial kernel. 

Our study focuses on two natural parameters for \knapsack: The number of different item weights $w_{\#}$, and the number of different item profits $p_{\#}$. Our main technical contribution is a proof showing that \knapsack\ does not admit a polynomial kernel for any of these two parameters under standard complexity-theoretic assumptions. Our proof discovers an elaborate application of the standard kernelization lower bound framework, and develops  along the way novel ideas that should be useful for other problems as well. We complement our lower bounds by showing the \knapsack\ admits a polynomial kernel for the combined parameter $w_{\#}+p_{\#}$.
\end{abstract}

\smallskip



\section{Introduction}

This paper proves a new complexity-theoretic barrier for the classical \knapsack\ problem. Namely, we prove that \knapsack\ has no polynomial kernels when parameterized by either the number~$w_\#$ of different weights, or the number $p_\#$ of different profits. Our results hold under the standard complexity-theoretic assumption \NPneqcoNPpoly. We also show that when both $w_\#$ and $p_\#$ are taken as a combined parameter, \knapsack\ does admit a polynomial kernel. Below we give a brief review of recent algorithmic progress for the \knapsack\ problem, as well as a quick survey through kernelization and parameterized complexity. We then describe how our results fit into the current state of the art, and give an overview of the main techniques used for obtaining our new hardness result for \knapsack.

\paragraph{The Knapsack problem.}

\knapsack\ (also known as \textsc{0/1 Knapsack}) is one the most fundamental and well studied problems in combinatorial optimization and theoretical computer science. In its most basic form, it is defined as follows:

\problemdef{\knapsack}
{A set $\mathcal{X} = \{x_1, \ldots, x_n\}$ of $n$ items, a weight function $w: \mathcal{X} \to \mathbb{N}$, a profit function $p: \mathcal{X} \to \mathbb{N}$, and two integers $W$ and $P$.}
{Is there a subset~$\mathcal{S} \subseteq \mathcal{X}$ with total weight $w(\mathcal{S})=\sum_{x \in \mathcal{S}} w(x) \leq W$ and total profit $p(\mathcal{S})=\sum_{x \in \mathcal{S}} p(x) \geq P$?}

\knapsack\ enjoys a key status in algorithmic design due to numerous reasons. First, it has many natural applications, in various practical areas such as resource allocation and scheduling. Second, it has immense educational value: Karp's \NP-hardness proof (from his seminal paper~\cite{Karp72}) is the first example of a reduction to a problem involving numbers, while the $O(Wn)$-time (or $O(Pn)$-time) algorithm by Bellman~\cite{bellman1957dynamic} from 1957 is one of the first dynamic programming algorithms, and it is still taught in most undergraduate algorithms courses to this day. And third, \knapsack\ has deep connections to other areas of computation: For example, one of the earliest cryptosystems by Merkle and Hellman~\cite{MH78} was based on \knapsack, and this was later extended to a host of other \emph{Knapsack-type} cryptosystems~\cite{BO88,CR88,IN96,Odl90}.

\knapsack\ is also important since it is both a generalization and a special case of a few other classical problems. For instance, it is a special case of the fundamental scheduling problem of minimizing the weighted number of tardy jobs on a single machine, the so-called $1||\sum w_jU_j$ problem~\cite{pinedo2016scheduling}. The variant of \knapsack\ where $W=P$ and $w(x)=p(x)$ for each item $x \in \mathcal{X}$ is precisely the \subsum\ problem: 

\problemdef{\subsum}
{A set $\mathcal A = \{a_1, \dots, a_{n}\}$ of $n$ non-negative integers and a target integer~$B$.}
{Is there a subset~$\mathcal A^* \subseteq \mathcal A$ with $\sum_{a\in \mathcal A^*} a =B$?}

Entire books \cite{KPP04book,MT90book} are dedicated to algorithmics for \knapsack. Quite surprisingly, major algorithmic advances are still being discovered in recent years. These typically involve improvements on Bellman's classical $O(\min \{W,P\} \cdot n)$-time algorithm in cases where either the maximum weight~$w_{\max} = \max_x w(x)$ or the maximum profit $p_{\max} = \max_x p(x)$ are relatively small.
Pisinger~\cite{Pisinger99} was the first to present such an improvement with his~$O({w_{\max} \cdot p_{\max} \cdot n})$-time algorithm. Later followed a series of papers with increasing improvements~\cite{AxiotisTzamos19,BateniHSS18,KellererP04}, culminating in an $\widetilde{O}(\min\{w^3_{\max}, p^3_{\max}\} + n)$-time\footnote{{We use $\widetilde{O}(\cdot)$ to suppress polylogarithmic factors.}} algorithm by Polak \etal~\cite{PolakEtAl21}, and an ${\widetilde{O}(\min\{w^{2/3}_{\max} \cdot p_{\max},\, w_{\max} \cdot p^{2/3}_{\max}\} \cdot n)}$-time algorithm by Bringmann and Cassis~\cite{BringmannCassis23}.
Very recently, Bringmann~\cite{Bringmann23} as well as Jin~\cite{Jin23} announced $\widetilde O(n + w_{\max}^2)$-time algorithms (and also ${\widetilde O(n+p_{\max}^2)}$-time algorithms) for \knapsack.
Cygan~\etal~\cite{CyganMWW19} showed that there are no~${O((W+n)^{2-\varepsilon})}$-time algorithms for \knapsack, for any~$\varepsilon > 0$, unless $(\min,+)$-convolutions can be solved in subquadratic time. An easy modification of their argument shows also that there are also no $O((P+n)^{2-\varepsilon})$-time algorithms under the same hypothesis. However, interestingly enough, these two lower bounds do not hold simultaneously together as shown by Bringmann and Cassis~\cite{BringmannCassis22}, who showed that \knapsack\ can be solved in $O((W+P)^{1.5}+n)$ time.

\paragraph{Parameters \boldmath{$w_{\#}$} and \boldmath{$p_{\#}$}.}

As discussed above, the time complexity of \knapsack\ with respect to $w_{\max}$ and $p_{\max}$ is well understood. Thus, it makes sense to consider other parameters. Let~$w_{\#}=|\{w(x_i):  i \in \{1, \ldots, n\}\}|$ denote the number of \emph{different} weights in a given \knapsack\ instance, and let $p_{\#}=|\{p(x_i): i \in \{1, \ldots, n\}|$ denote the number of different profits. These parameters can be expected to be small in several natural applications, and have been initially studied for \subsum\ (i.e.,. when $w(x_i)=p(x_i)$ for all items $x_i$) by Fellows \etal~\cite{FellowsGR12} in the context of the more general \emph{small number of numbers} paradigm. Clearly, $w_{\#} \leq w_{\max}$ and $p_{\#} \leq p_{\max}$, and so designing algorithms which are efficient in terms of $w_{\#}$ or~$p_{\#}$ is more challenging. In particular, since \knapsack\ is \NP-hard and $w_{\#},p_{\#} \leq n$, we cannot expect algorithms with running times of the form~$(w_{\#} \cdot n)^{O(1)}$ or $(p_{\#} \cdot n)^{O(1)}$.

Since polynomial-time algorithms with respect to $w_{\#}$ and $p_{\#}$ are unlikely, it is natural to consider these two parameters in the context of parameterized complexity~\cite{CyganFKLMPPS15}. In this context, it is not difficult to show \knapsack\ can be formulated as an Integer Linear Program (ILP) with either $O(w_{\#})$ or $O(p_{\#})$ many variables. Using one of several solvers for ILP with few variables, such as Lenstra's famous algorithm~\cite{Lenstra83,Kannan87},  gives us algorithms with running times of $2^{\widetilde{O}(w_{\#})} \cdot |I|$ and~$2^{\widetilde{O}(p_{\#})} \cdot |I|$ for \knapsack, where $|I|$ is the total encoding length of the input (see Hermelin \etal~\cite{HermelinKPS21} for algorithms with similar running times for the more general $1||\sum w_jU_j$ problem). 
Such algorithms are known as \emph{fixed-parameter algorithms} (FPT) in the terminology of parameterized complexity.

Note that both of these algorithms cannot be significantly improved assuming the Exponential Time Hypothesis (ETH). Indeed, it is known that assuming ETH, there are no $2^{o(n)}$-time algorithms for \subsum~\cite{BuhrmanLT15,JansenLL16}. As both $w_{\#}$ and $p_{\#}$ are bounded by~$n$, and \subsum\ is a special case of \knapsack, this implies that there are most likely no $2^{o(w_{\#})} \cdot n$ or $2^{o(p_{\#})} \cdot n$ time algorithms for \knapsack.
Using the Strong Exponential Time Hypothesis (SETH) instead of ETH, Abboud \etal~\cite{AbboudBHS22} showed a stronger lower bound excluding $B^{1-\epsilon} \cdot 2^{o(n)}$-time algorithms for any~${\epsilon> 0}$ (which excludes $(W + P)^{1-\epsilon} \cdot 2^{o(w_{\#} + p_{\#})}$-time algorithms for \knapsack).
Once fixed-parameter algorithms for a particular problem---in this case, \knapsack{}---have been devised, the next step is to understand the kernelization complexity of the problem at hand.

\paragraph{Kernelization.} 

One of the most fundamental and important techniques in parameterized complexity is kernelization~\cite{kernelization_book}:
\begin{definition}
A \emph{kernelization} algorithm (or \emph{kernel}) for a parameterized problem $\Pi$ is an algorithm that receives as input an instance~$I$ of~$\Pi$ with parameter~$k$ and outputs in polynomial time an instance~$J$ of~$\Pi$ with parameter~$\ell$ such that
\begin{itemize}
\item $(I,k)$ is a ``yes"-instance for $\Pi$ if and only if $(J,\ell)$ is ``yes"-instance for $\Pi$, and
\item $|J| + \ell \leq f(k)$ for some computable function $f(\cdot)$.
\end{itemize}
The function $f(\cdot)$ is referred to as the \emph{size} of the kernel. 
\end{definition}

Thus, a kernel is a self-reduction from a problem onto itself that produces an equivalent instance with size bounded by the input parameter. In this way, kernelization may be thought of as preprocessing that aims to simplify or ``kernelize" a problem instance by reducing its size while preserving some solution. Following this line of thought, problems that admit kernels of small sizes may be thought of as problems that allow efficient and effective preprocessing. For this reason, research into kernelization algorithms has seen a significant surge in recent years, and has become one of the central topics in parameterized complexity (see e.g. the monograph by Fomin \etal~\cite{kernelization_book} and the numerous references within). 

In the context of \knapsack, a kernelization algorithm for say parameter~$w_{\#}$ transforms any \knapsack\ instance in polynomial-time into an equivalent instance where the total encoding length of item weights and profits are bounded by~$f(w_{\#})$, for some function $f(\cdot)$. Observe that the $2^{\widetilde O(w_{\#})} \cdot n$-time algorithm for \knapsack\ implies a kernel of size $2^{\widetilde O(w_{\#})}$ for the problem: Indeed, let $|I|$ denote the total encoding length of item weights and profits in a given \knapsack\ instance and observe that $n \leq |I|$. The kernelization algorithm can first check whether $|I| = 2^{\widetilde O(w_{\#})}$. If this is the case, then the instance already has size bounded by $2^{\widetilde O(w_{\#})}$, and otherwise it is solvable in polynomial time by the $2^{\widetilde O(w_{\#})} \cdot n=2^{\widetilde O(w_{\#})} \cdot |I|$-time algorithm. A similar argument shows that \knapsack\ has a kernel of size $2^{\widetilde O(p_{\#})}$.

The obvious question to ask is whether we can obtain smaller kernels with respect to either~$w_{\#}$ or~$p_{\#}$. Here, the gold standard in parameterized complexity are \emph{polynomial kernels}, kernels with size $f(k) = k^{O(1)}$. By now we know of countless fixed-parameter tractable (\NP-hard) problems that also admit polynomial kernels~\cite{kernelization_book}; yet, at the same time, for many other problems polynomial kernels were shown to be unlikely to exist. Thus, the central question this paper addresses is: 
\begin{quote}
Does \knapsack\ admit a polynomial kernel with respect to either $w_{\#}$ or~$p_{\#}$?     
\end{quote}
Note that there are previously known results that give encouraging indications to to the question above. Etscheid \etal~\cite{EtscheidKMR17} show that \subsum\ admits a polynomial kernel with respect to~$a_{\#}$, the number of different numbers in the instance (i.e, $a_{\#}=|\{a_i: i \in \{1,\ldots, n\}\}|$). Can this result be generalized to hold also for the \knapsack\ problem?

\subsection{Our results}

Our main technical result of the paper is a negative answer to the question above. We use the by now standard framework for excluding polynomial kernels~\cite{BodlaenderDFH09} based on the assumption \NPneqcoNPpoly\ (which implies that the polynomial hierarchy collapses) to show the following:
\begin{theorem}
\label{thm:main}%
Assuming \NPneqcoNPpoly, there is no polynomial kernel for \knapsack\ parameterized by the number $w_{\#}$ of different weights, nor by the number $p_{\#}$ of different profits.
\end{theorem}
\noindent The proof of the theorem above is obtained through a rather involved construction of what is known as a composition algorithm (see \Cref{def:composition}). This algorithm composes several instances of a specialized variant of \subsum\ into a single \knapsack\ instance where either $w_\#$ or~$p_\#$ is kept relatively small. We give an overview of this algorithm in \Cref{sec:overview}. 

Complementing the negative result of \Cref{thm:main}, we show that when both $w_{\#}$ and $p_{\#}$ are taken as a combined parameter, the polynomial kernel for \subsum~\cite{EtscheidKMR17} can be extended to a polynomial kernel for \knapsack. 
\begin{theorem}
\label{thm:secondary}%
\knapsack\ parameterized by $w_{\#} + p_{\#}$ admits a polynomial kernel.
\end{theorem}
\noindent Thus, the lower bounds for $w_{\#}$ and  $p_{\#}$ cannot be combined. This is somewhat reminiscent to the situation mentioned above, where \knapsack\ is unlikely to have algorithms with subquadratic running times in either~$(W+n)$ or~$(P+n)$~\cite{CyganMWW19}, but admits an algorithm with with subquadratic running time in~$(W+P+n)$~\cite{BringmannCassis22}.   

\subsection{Technical overview}
\label{sec:overview}%

We prove \Cref{thm:main} for the parameter~$w_{\#}$; the statement for parameter~$p_{\#}$ then follows by a reduction from Polak \etal~\cite[Chapter~4]{PolakEtAl21}. In broad terms, our proof follows the standard framework for showing kernelization lower bounds that has been developed over the years~\mbox{\cite{BodlaenderDFH09,BodlaenderJK14,DellM14,Drucker15}}. In this framework, to exclude a polynomial kernel for a given problem $\Pi_1$ parameterized by some parameter $k$, one shows what is called a \emph{composition algorithm} (see~\Cref{def:composition} for a formal definition) from some known \NP-hard problem $\Pi_2$ to $\Pi_1$. This algorithm takes as input $t$ instances of~$\Pi_2$ of size $n$ each, and converts these instances in polynomial time to an equivalent instance of~$\Pi_1$ such that $k = (n \lg t)^{O(1)}$. Here equivalence means that at least one of the $t$ input instances is a ``yes"-instance for $\Pi_2$ if and only if the output instance of the composition is a ``yes"-instance for~$\Pi_1$. Such a composition algorithm then directly implies that $\Pi_2$ has no polynomial kernel with respect to~$k$ under the assumption \NPneqcoNPpoly~\cite{BodlaenderDFH09,BodlaenderJK14}.

We show a composition algorithm from a restricted version of \subsum\ which we call \rss. In this restricted version, any instance of size $n$ is restricted to include only numbers from a known set $\widetilde{\mathcal{A}}_n$ of size $O(n^3)$. This allows us to bound the number of different numbers in any $t$ instances of \rss\ of the same size. Our composition algorithm then proceeds as follows: Given any $t$ instances $\mathcal A_0,\ldots,\mathcal A_{t-1}$ of \rss\, the algorithm converts any integer $a \in \mathcal A_i$, $i \in \{0,\ldots,t-1\}$, to a \knapsack\ item~$x_a$ (which we refer to as an \emph{encoding item}) with weight~$w(x_a)=a$. We then assign a higher profit to items corresponding to \rss\ instances of higher index, meaning that it is always less profitable to choose an item corresponding to some integer $a \in \mathcal A_i$ than choosing an item corresponding to some integer $a \in\mathcal  A_{i_0}$ for $i < i_0$. This allows us to encode any solution $\mathcal A^*_i \subseteq \mathcal A_i$ to the $i$'th \rss\ instance by a solution $\mathcal{X}(\mathcal A^*_i)$ to our \knapsack\ instance which includes all items $x_a$ for $a \in \mathcal A^*_i$, and all items $x_a$ for $a \in \mathcal A_{i_0}$ for $i < i_0$. Thus, if we know index~$i$ in advance, our composition would be complete.

However, we do not have a priori knowledge of index~$i$. To circumvent this, we use \emph{index items} that encode the selection of~$i \in \{0,\ldots,t-1\}$. These are $2 \lg t$ items that encoding the selection of $t$ binary values $i(0),\ldots,i(\lg t -1)$ that correspond to the base 2 representation of~$i$, \ie $i = \sum_k i(k) \cdot2^k$. This is somewhat akin to the \emph{``colors and IDs''} technique for composition algorithms which was introduced by Dom \etal~\cite{DomLS14}. The goal of these index items is to ensure that choosing a solution $\mathcal{X}(\mathcal A^*_i)$, for any value of~$i$, will allow adding a subset of index items to the knapsack such that any such selection will have the same weight and profit. The obstacle is that the difference in profit between solutions $\mathcal{X}(\mathcal A^*_0)$ which includes almost all items, and solutions~$\mathcal{X}(\mathcal A^*_{t-1})$ which includes only items corresponding to instance $\mathcal A_{t-1}$, is quadratic in~$i$ (or~$t$). Hence, we cannot compensate for this difference using only $O(\lg t)$ index items.

We therefore introduce an additional $O(\lg^2 t)$ items which we refer to as  \emph{quadratization items}. These encode the selection of binary pair values $i(k) i(\ell)$ in $i^2 = \sum_k \sum_\ell i(k) i(\ell) \cdot 2^{k+\ell}$, and allow us to encode the quadratic compensation mentioned above. Unfortunately, this introduces additional technical difficulties, such as ensuring compatibility between the selection of the index and quadratization items. Meanwhile, we still need to maintain the compatibility between the encoding and index items as well. These difficulties are overcome using various applications of a basic algebraic lemma that we prove in \Cref{sec:preliminaries}. In the end, we obtain an instance with $w_{\#}=O(n^3 \cdot \lg^2 t)$ many different weights, which by previously known results (\Cref{thm:NoPolyKernel}) implies that \NPneqcoNPpoly. The full details of the entire composition are given in \Cref{sec:wsharp}.


\section{Preliminaries}
\label{sec:preliminaries}%

We next quickly review the kernelization lower bounds framework, introduced by Bodlaender \etal~\cite{BodlaenderDFH09} and further developed by~\cite{BodlaenderJK14,DellM14,Drucker15}, that will be used for proving \Cref{thm:main}. At the heart of the framework lies the notion of a composition.
\begin{definition}
\label{def:composition}
A composition algorithm from a problem $\Pi_1$ to a parameterized problem $\Pi_2$ is an algorithm that receives as input $t$ instances $I_0, \ldots, I_{t-1}$ of size~$n$ of~$\Pi_1$, and computes in polynomial time an instance~$(J,k)$ of $\Pi_2$ such that
\begin{itemize}
\item $J$ is a ``yes"-instance of $\Pi_2$ if and only if $I_i$ is a ``yes"-instance of $\Pi_1$ for some $i \in \{0, \ldots, t-1\}$,
\item and $k \leq (n + \lg t)^{O(1)}$.
\end{itemize}     
\end{definition}

The main connection between composition algorithms and the exclusion of polynomial kernels is given in the theorem below, whose proof relies on a complexity-theoretic lemma by Fortnow and Santhanam~\cite{FortnowSanthanam11}.
\begin{theorem}[\cite{BodlaenderDFH09,BodlaenderJK14}]
\label{thm:NoPolyKernel}%
Let $\Pi_1$ be an \NP-hard problem, and $\Pi_2$ be a parameterized problem with a composition algorithm from~$\Pi_1$ to~$\Pi_2$. Then $\Pi_2$ does not admit a polynomial kernel assuming \NPneqcoNPpoly.
\end{theorem}
\noindent Note that \NPneqcoNPpoly\ is a widely believed assumption in complexity theory, and if it were false then the polynomial hierarchy would collapse to its third level~\cite{Yap83}.

As a final point, we will also make use of the following basic algebraic lemma throughout our proof. Below we provide a proof for the sake of completeness. 
\begin{lemma}
\label{lem:algebraic}
Let $b > 1$ and $k$ be fixed positive integers. Then there exists a unique integer solution to the equation
\begin{equation*}
x_0b^0 + x_1b^1 + \cdots + x_{k-1}b^{k-1} = \sum^{k-1}_{i=0} b^i
\end{equation*}
constrained by $x_i \in \{0,1, \ldots,b\}$ for each $i \in \{0,\ldots,k-1\}$. Namely, the solution is $$x_0 = x_1 = \cdots = x_{k-1} = 1\enspace.$$
\end{lemma}

\begin{proof}
    We show the statement by induction on $k$.
    For $k = 1$, the statement is obvious.
    So fix~${k > 1}$.
    Note that $\sum_{i=0}^{k-1} b^i = b^{k-1} + \sum_{i=0}^{k-2} b^i = b^{k-1} + \frac{b^{k-1}}{b-1} < 2b^{k-1}$, so we have that $x_{k - 1} \in \{0, 1\}$.
    Further, we have $\sum_{i=0}^{k-2} x_i b^i \le \sum_{i=0}^{k-2} b \cdot b^i = \sum_{i=1}^{k-1} b^i < \sum_{i=0}^{k-1} b^i$, implying that $x_k > 0$.
    Thus, we have $x_{k-1} =1$ and $x_i = 1$ for $i \in \{0, \ldots, k-2\}$ follows by induction.
\end{proof}


\section{No Polynomial Kernel for Parameter \boldmath{$w_{\#}$}}
\label{sec:wsharp}

In the following, we present the proof of \Cref{thm:main} for parameter~$w_{\#}$, the total number of different weights in a \knapsack\ instance. As mentioned above, in our proof we construct a composition from a restricted version of \subsum\ to \knapsack\ parameterized by~$w_{\#}$. The proof is divided into four parts: In the first, we introduce the \rss\ problem, and prove that it is \NP-hard. In the second part we begin to describe our composition by showing how to encode $t$ instances $\mathcal A_0,\ldots,\mathcal A_{t-1}$ of \rss\ into a single set $\mathcal{X}$ of \emph{encoding items}. In the third part we describe the set $\mathcal{Y}$ of \emph{quadratization items}, and the set $\mathcal{Z}$ of \emph{index items}, which together form the instance selection gadget of the composition. The final part is devoted to finishing details, and to proving the correctness of the composition.  

\subsection{Restricted Subset Sum}

We start with a restricted version of \subsum\ which allows us to bound the number of different numbers appearing in any set of $t$ \subsum\ instances. Namely, in our restricted version of \subsum, any instance of size $n$ contains numbers from a restricted set of $O(n^3)$ many numbers. Apart from this, its useful properties are that the target value only depends on the number of input numbers, and that any solution must have the same cardinality. For this, we define the set of possible numbers which may be contained in a \rss\ instance of size $n$ as
$$
\widetilde{\mathcal{A}}_n := \big\{(3n+1)^{j_1} + (3n+1)^{j_2} + (3n+1)^{j_3} \,\, : \,\, j_1, j_2, j_3 \in \{1,\ldots,3n\}\big\}\enspace .
$$ 
Note that $\widetilde{\mathcal{A}}_n$ contains $27n^3=O(n^3)$ many integers. Furthermore, we also 
define a global target for all instances of size $n$ by~$B_n := \sum_{j=1}^{3n} (3n+1)^j$.

\problemdef{\rss}
{A set $\mathcal A = \{a_1, \dots, a_{3n}\}$ of $3n$ integers from $\widetilde{\mathcal{A}}_n$ with $\sum_{j=1}^{3n} a_j = 3B_n$.}
{Is there a subset~$\mathcal A^* \subseteq \mathcal A$ with $|\mathcal A^*|=n$ such that $\sum_{a \in \mathcal A^*} a = B_n$?}

\begin{lemma}
\label{lem:NP-hard}%
\rss\ is \NP-complete.
\end{lemma}

\begin{proof}
We present a reduction from a variant of \rxc\ where each element appears in exactly three sets. Recall that in \rxc, the input consists of a set $\mathcal{T}$ consisting of 3-element subsets of~$\{1,\ldots,3n\}$ such that each~$j \in \{1,\ldots,3n\}$ appears in exactly three sets from~$\mathcal{T}$, and the question is whether there exists a subset~$\mathcal{T}' \subseteq \mathcal{T}$ such that each element from~$\{1,\ldots,3n\}$ appears in exactly one set from~$\mathcal{T}'$. This variant of \rxc\ is well-known to be \NP-hard~\cite{Gonzalez85}.

Our reduction is almost identical the original hardness reduction for \subsum\ by Karp~\cite{Karp72}, only that we reduce from \rxc\ instead of the more general \textsc{Exact Cover}: Let $\mathcal{T}$ be an instance of \rxc. For each 3-element set~$T \in \mathcal{T}$, a number~$a_T := \sum_{j \in T} (3n+1)^j$ is added to the \rss\ instance. In this way, the constructed instance of \rss\ is $\mathcal A :=\{a_T : T \in  \mathcal{T}\}$.
Note that $a_T \in \widetilde{\mathcal{A}}_n$ for every $T \in \mathcal{T}$.
Further, since each~$j \in \{1,\ldots,3n\}$ appears in exactly three sets from~$\mathcal{T}$, we have $\sum_{T \in \mathcal{T}} a_T = 3 \sum_{j=1}^{3n} (3n+1)^j = 3 B_n$.
Below we prove the correctness of this construction.
    
Suppose there is a solution~$\mathcal{T}' \subseteq \mathcal{T}$ for the \rxc\ instance. Then $|\mathcal{T}'| = n$ and we have $\sum_{T \in \mathcal{T}} a_T = \sum_{\ell=1}^{3n} (3n+1)^\ell=B_n$. Conversely, suppose there is a solution~$\mathcal A^* \subseteq \mathcal A$ with $|\mathcal A^*| =n$ for the \rss\ instance. 
Let \mbox{$\mathcal{T}^* := \{T : a_T \in \mathcal A^*\}$}.
Since each term $(3n+1)^{j}$ appears only $3< 3n+1$ times, \Cref{lem:algebraic} implies that the only way for some numbers from~$\mathcal A$ to add up to~$B_n = \sum_{j=1}^{3n} (3n+1)^j$ is that each term $(3n+1)^j$ appears in exactly one~$a_T \in \mathcal A^*$.
In other words, for each $j \in \{1,\ldots,3n\}$, there is exactly one $T \in \mathcal{T}^*$ with~$j \in T$.
\end{proof}

\subsection{Encoding Gadget}
\label{subsec:encoding}

In the following, we show how to encode $t$ instances of \rss\ into a single \knapsack\ instance. Throughout the remainder of the section, we use $\mathcal A_0, \ldots , \mathcal A_{t-1}$ to denote the $t$ input instances of \rss\ to our composition, where $|\mathcal A_i|=3n$ for each~$i\in \{0,\ldots,t-1\}$. By copying instances, we may assume without loss of generality that~${t = 2^k}$ for some~$k \in \mathbb{N}$. Furthermore, we let $\mathcal A_i = \{a^i_1,\ldots,a^i_{3n}\}$ denote the $i$th instance for each~$i\in{\{0,\ldots,t-1\}}$. By the definition of \rss, we have $a^i_j \in \widetilde{\mathcal{A}}_n$ for each each~$i\in\{0,\ldots,t-1\}$ and~$j\in\{1,\ldots,3n\}$, and~$\sum_j a^i_j = 3B_n$ for each~$i\in\{0,\ldots,t-1\}$.

Recall that \subsum\ can be seen as a special case of \knapsack\ where the profit of each item equals its weight. This yields an easy reduction from each single \rss\ instance~$\mathcal A_i$ to \knapsack. We apply this reduction with a slight modification: To capture the condition that each solution of \rss\ shall contain exactly $n$ numbers, we add a large number $X = 3tn B_n$ to each number in a \rss\ instance, and set~${B := B_n + n \cdot X}$. We also increase the profit of each item corresponding to the~$i$'th \rss\ instance by adding~$i \cdot 3B$ to its profit. More precisely, for each~$i \in \{0,\ldots,t-1\}$ and~$j \in \{1,\ldots,3n\}$, we construct an \emph{encoding item} $x^i_j$ with 
\begin{itemize}
\item $w(x^i_j) = X + a^i_j$, and
\item $p(x_i^j) = X+ a^i_j + i \cdot 3B$.
\end{itemize}  
We use $\mathcal{X} = \{x^i_j: i \in \{0,\ldots,t-1\}, j \in \{1,\ldots,3n\}\}$ to denote the set of all encoding items.

Let $\mathcal{X}_i$ denote the set of encoding items corresponding to instance $\mathcal A_i$ of \rss, \ie $\mathcal{X}_i = \{x_j^i : j\in \{1,\ldots,3n\}\}$. Note that $w(\mathcal{X}_i)  = 3B $ and $p(\mathcal{X}_i) = 3B  + 9nB \cdot i$. Now suppose $\mathcal A_i$ has a solution $\mathcal A^*_i \subset \mathcal A_i$. We would like to encode this solution using the following set of encoding items
$$
\mathcal{X}(\mathcal A^*_i) = \mathcal{X}^*_i \cup \mathcal{X}_{i+1} \cup \mathcal{X}_{i+2} \cup \cdots \cup \mathcal{X}_{t-1},
$$
where $\mathcal X_i^* := \{x_i^j : a_i^j \in \mathcal A_i^* \} $ is the set of items corresponding to~$\mathcal A_i^*$. An easy calculation shows that if the elements of $\mathcal A^*_i$ sum up to~$B$ (\ie $\mathcal A^*_i$ is indeed a solution), then $w(\mathcal{X}(\mathcal A^*_i)) = (3t-3i-2)\cdot B$ and $p(\mathcal{X}(\mathcal A^*_i)) = (3t-3i-2)\cdot B + (\binom{t}{2}-\binom{i+1}{2}+\frac{i}{3}) \cdot 9nB$. The next lemma shows that the converse is also true; namely, that if there is a subset of encoding items with the above weight and profit, then there must be a solution to the $i$th \rss\ instance.

\begin{lemma}
\label{lem:encoding}%
Let $i \in \{0,\ldots,t-1\}$. There exists a subset $\mathcal{X}^* \subseteq \mathcal{X}$ of encoding items with total weight 
$$
w(\mathcal{X}^*) \leq (3t-3i-2) \cdot B
$$
and total profit 
$$
p(\mathcal{X}^*)\geq (3t-3i-2) \cdot B \,\,+\,\,  \Bigr(\binom{t}{2} - \binom{i+1}{2} + \frac{i}{3}\Bigr) \cdot 9nB 
$$
for our \knapsack\ instance if and only if there exist a solution $\mathcal A^*_i \subseteq \mathcal A_i$ to the $i$'th \rss\ instance.
\end{lemma}

\begin{proof}
As mentioned above, if $\mathcal A^*_i \subseteq \mathcal A_i$ is a solution to the $i$'th \rss\ instance, then the set of encoding items~$\mathcal{X}(\mathcal A^*_i)$ has total weight and profit satisfying the bounds in the lemma. We proceed then to prove the converse direction, namely that a set of items with weight and profit bounded as above implies that instance $\mathcal A_i$ has a solution.

Let $\mathcal{X^*}$ be any subset of encoding items with $w(\mathcal{X}^*) \leq (3t-3i-2) \cdot B$. As $B= B_n+ nX$ and $B_n = X/(3nt)$, we have 
\begin{align*}
w(\mathcal{X}^*) &\le  (3t-3i-2) \cdot B \\
& = (3t-3i-2) \cdot X/(3nt) + (3t-3i-2) \cdot nX \\
& < X + (3t-3i-2) \cdot nX \\
& = X + (n+3(t-i-1)n) \cdot X\enspace.
\end{align*}
Thus, since each encoding item has weight at least $X$, the inequality above shows that $\mathcal{X^*}$ contains at most $n + 3(t-1-i) n$ elements.

Next assume $p(\mathcal{X^*}) \geq (3t-3i-2)\cdot B + 9n (\binom{t}{2}-\binom{i+1}{2}+\frac{i}{3}) \cdot B$. We show that as ${|\mathcal{X^*}| \leq n + 3(t-1-i) n}$, it must be that $\mathcal{X^*}$ contains up to~$n$ items of $\mathcal{X}_i$, and all $3(t-1-i) n$~items of $\mathcal{X}_{i+1} \cup \cdots \cup \mathcal{X}_{t-1}$. That is, $\mathcal{X}^*=\mathcal{X}(\mathcal A^*_i)$ for some subset of $n$ elements $\mathcal A^*_i \subset \mathcal A_i$. Using the facts that $p(x_j^i) \le (1 + 1/(nt))X + i \cdot 3B$ for every $j \in \{1,\ldots,3n\}$ and that $\sum_{a\in \mathcal A_{i_0}} a = 3\cdot B_n$ for every $i_0 \in \{0, \ldots, t-1\}$, we observe that taking any subset of~$n$ elements~$\mathcal A^*_i \subset \mathcal A_i$ gives us a set $\mathcal{X}(\mathcal A^*_i)$ with $n + 3(t-1-i) n$ encoding items that fulfills 
\begin{align*}
p(\mathcal{X}(\mathcal A^*_i)) & \le n \cdot ((1+\frac{1}{tn})X + i \cdot 3B) + \sum_{i_0 = i + 1}^{t-1} \Bigl(3B_n + 3n \cdot (X + i_0 \cdot 3B)\Bigr)\\
& = nX + \frac{X}{t} + i \cdot 3nB + 3(t-1-i) \cdot B + \sum_{i_0 = i + 1}^{t-1} i_0 \cdot 9nB\\
& < B + \frac{X}{t} + (3t - 3i -3) \cdot B +  \Bigl(\binom{t}{2} - \binom{i+1}{2}+ \frac{i}{3}\Bigr) \cdot 9nB\\
& = \frac{X}{t} + (3t -3i - 2) \cdot B  +  \Bigl(\binom{t}{2} - \binom{i+1}{2}+ \frac{i}{3}\Bigr) \cdot 9nB\enspace.
\end{align*}
Replacing any item from~$\mathcal{X}_{i_1}$ with an item~$a_{j_0}^{i_0}$ from~$\mathcal{X}_{i_0}$ for $i_0 < i \leq i_1$, decreases the profit by $3\cdot (i_1 -i_0) \cdot B - a_{j_0}^{i_0} > X/t$. Therefore, any set of items arising through such replacements cannot reach the desired profit.
Further, taking less than $n$ items from~$\mathcal{X}_i$ reduces the desired profit by at least $X > X/t$.
It follows that $\mathcal{X}=\mathcal{X}(\mathcal A^*_i)$ for some subset $\mathcal A^*_i \subset \mathcal A_i$ of up to~$n$ elements. 

To complete the proof, we show that $\mathcal A_i^*$ is indeed a solution to~$\mathcal A_i$. That is, all elements of~$\mathcal A_i^*$ sum up to~$B_n$.  First note that
\begin{align*}
B + (t-1-i)\cdot 3B & \ge w(\mathcal{X})\\
&= w(\mathcal X_i^*) + \sum_{i_0 = i + 1}^{t-1} w(\mathcal X_i)\\
& = \sum_{a \in \mathcal A_i^*} a + n X + \sum_{i_0 =i+1}^{t-1} \Bigl(3B_n + 3n \cdot X\Bigr)\\
& = \sum_{a \in \mathcal A_i^*} a + n \cdot X + (t-1-i) \cdot 3B,
\end{align*}
and so~$\sum_{a\in \mathcal A_i^*} a \le B- n \cdot X = B_n$. Next observe that 
\begin{align*}
B + & (t-1-i) \cdot 3B  +  \Bigr(\binom{t}{2} - \binom{i+1}{2} + \frac{i}{3}\Bigr) \cdot 9nB  \le p(\mathcal{X}) \\ 
&= p(\mathcal X_i^*) + \sum_{i_0 = i + 1}^{t-1} p(\mathcal X_i)\\
& = \sum_{ a\in\mathcal  A_i^*} a + n \cdot (X + i \cdot 3B) + \sum_{i_0 = i+1}^{t-1} \Bigl( 3B + 3n \cdot i_0 \cdot 3B\Bigr)\\
& = \sum_{ a\in \mathcal A_i^*} a + nX + i \cdot 3nB + (t-1-i) \cdot 3B+ \Bigl( \binom{t}{2} - \binom{i+1}{2}\Bigr) \cdot 9nB\\
& = \sum_{ a\in \mathcal A_i^*} a + nX + (t-1-i) \cdot 3B+ \Bigl( \binom{t}{2} - \binom{i+1}{2} + \frac{i}{3}\Bigr) \cdot 9nB,
\end{align*}
and so~$\sum_{a\in \mathcal A_i^*} a \ge B - nX = B_n$. It follows that~$\sum_{a\in \mathcal A_i^*} a = B_n$, and so $\mathcal A^*_i$ is indeed a solution to instance~$\mathcal A_i$.
\end{proof}

\Cref{lem:encoding} implies that if we knew which \rss\ instance $\mathcal A_i$ has a solution, then we could easily set the weight and profit of our composed \knapsack\ instance to encode this solution. However, we do not have prior information about index~$i$. Furthermore, observe that as the value of~$i$ increases, both the profit and weight of the required solution decrease.  Since we do not know the value of~$i$, it would be beneficial to balance all possible choices of~$i$ in terms of weight and profit. 

Thus, the remaining construction focuses on ensuring that solutions to the \knapsack\ instance corresponding to different $\mathcal A_i$'s will all have the same weight and profit. In particular, the construction will guarantee that any choice of $i$ can obtain a profit of $(3t-2) \cdot B + 9  \cdot \binom{t}{2} \cdot nB$,
in addition to some large constant. Considering the profit guaranteed by solutions of \Cref{lem:encoding}, we need to compensate for the loss of the quadratic term 
\begin{equation}
\label{eqn:compensation}%
\Big(\binom{i+1}{2} - \frac{i}{3}\Big) \cdot 9nB\enspace.    
\end{equation}
We call the term above the \emph{compensation term of $i$}. It will play an important role in the remainder of our construction.

\subsection{Instance Selection Gadget}

We next add additional items to our \knapsack\ instance that will serve as an instance selection gadget. This gadget selects an instance of \rss\ for which presumably there is a solution. The gadget consists of two types of items: The \emph{index items} which encode an index of an \rss\ instance~$i \in \{0,\ldots,t-1\}$, and \emph{quadratization items} that help to encode the compensation term of $i$ given in~\autoref{eqn:compensation}.

\paragraph*{Quadratization items.}

The main idea behind the quadratization items is as follows: Any integer~$i\in\{0,\ldots,t-1\}$ can be written as the sum
$$
i = \sum_{k=0}^{\lg t - 1} i(k) \cdot 2^k,
$$
for some binary values $i(0),\ldots,i(\lg t-1) \in \{0,1\}$. Thus, using these same $\lg t$ binary values, we can write the compensation term of $i$ as 
\begin{equation}
\label{eqn:quadratization}%
\begin{split}
\Big( \binom{i+1}{2} - \frac{i}{3} \Big) \cdot 9nB \quad &=\quad  \left( 0.5 \cdot i^2 + \frac{1}{6} \cdot i \right) \cdot 9nB \\
&=\quad \left( 0.5 \cdot \sum_{k=0}^{\lg t - 1} \sum_{\ell=0}^{\lg t - 1} i(k) \cdot i(\ell) \cdot 2^{k+\ell} + \frac{1}{6} \sum_{k =0}^{\lg t-1} i(k)\cdot 2^k \right) \cdot 9nB \\
&=\quad  \Big( 9 \cdot \sum_{\substack{i(k)=1,\\ i(\ell)=1,\\ k < \ell}} 2^{k+\ell}  \,\, +\,\,   4.5 \cdot \sum_{i(k)=1} 2^{k + k}  \,\, +\,\, 1.5 \cdot \sum_{i(k)=1} 2^k \Big) \cdot nB\enspace.
\end{split}
\end{equation}
Thus, we construct $3\cdot \binom{\lg t}{2}+\lg t$ different quadratization items, each modeling the contribution of all possible values of $i(k)$ and $i(\ell)$,  $k \leq \ell \in \{0,\ldots,\lg t -1\}$, in the last equality of \autoref{eqn:quadratization} above.  

Let $Y=t^2 \cdot 3nB$, and observe that $Y$ is larger than the total profit of all encoding items. Furthermore, let $f: \{0,\ldots, \lg t -1\}^2 \to \{0,\ldots,\lg^2 t -1\}$ be any bijective function. For each pair of indices~$k$ and $\ell$ with $0 \leq k < \ell \leq \lg t -1$, we add three quadratization items~$y^{1,0}_{k,\ell}$, $y^{0,1}_{k,\ell}$, and~$y^{1,1}_{k,\ell}$ with the following weight and profit:
\begin{itemize}
\item $w(y^{1,0}_{k,\ell}) = p(y^{1,0}_{k,\ell}) = 3^{f(k,\ell)} \cdot Y$.
\item $w(y^{0,1}_{k,\ell}) = p(y^{0,1}_{k,\ell}) = 3^{f(\ell,k)} \cdot Y$. 
\item $w(y^{1,1}_{k,\ell}) = (3^{f(k,\ell)}+3^{f(\ell,k)}) \cdot Y$ and $p(y^{1,1}_{k,\ell}) =  (3^{f(k,\ell)}+3^{f(\ell,k)}) \cdot Y  + 2^{k+\ell} \cdot 9nB$.
\end{itemize}  
Furthermore, for each $k \in \{0,\ldots,\lg t-1\}$, we add a single quadratization item $y^{1,1}_{k,k}$ with:
\begin{itemize}
\item $w(y^{1,1}_{k,k}) = 3^{f(k,k)} \cdot Y$ and $p(y^{1,1}_{k,k}) = 3^{f(k,k)} \cdot Y + 2^{k + k} \cdot 4.5nB + 2^k \cdot 1.5 n B$.
\end{itemize}  
We use $\mathcal{Y}=\{y^{1,0}_{k,\ell}, y^{1,0}_{k,\ell}, y^{1,0}_{k,\ell} : 0 \leq k < \ell \leq \lg t -1\} \cup \{y^{1,1}_{k,k}: 0 \leq k \leq \lg t -1\}$ to denote the set of all quadratization items.

The additional terms that depend on $Y$ will become clearer when we introduce the index items. But for now, one can observe that the smaller terms used in the profits of $y^{1,1}_{k,\ell}$ and $y^{1,1}_{k,k}$ allow us to encode the compensation term of~$i$. In particular, an easy calculation using \autoref{eqn:quadratization} gives us the following useful lemma:
\begin{lemma}
\label{lem:quadratization}
Let $i \in \{0,\ldots,t-1\}$, and let $i(0),\ldots,i(t-1) \in \{0,1\}$ be binary values such that $i = \sum_k i(k) \cdot 2^k$. Moreover, let $\mathcal{Y}_i$ denote the set of quadratization items defined by
$$
\mathcal{Y}_i = \{ y^{i(k),i(\ell)}_{k,\ell} : 0 \leq k \leq \ell \leq \lg t -1 \},
$$
where $y^{i(k),i(\ell)}_{k,\ell}$ is the empty item (\ie an item with weight and profit 0) if $i (k)= i (\ell)=0$. Then 
$$
p(\mathcal{Y}_i) \quad = \quad w(\mathcal{Y}_i)+ \Big( \binom{i + 1}{2} - \frac{i}{3} \Big) \cdot 9nB\enspace.
$$
\end{lemma}

\paragraph*{Index items.}

The index items ensure that only quadratization items that correspond to subsets~$\mathcal{Y}_i$ as in \Cref{lem:quadratization} above can be picked into any solution of our \knapsack\ instance. In particular, the index items will encode the selection of an index~$i \in \{0,\ldots,\lg t -1\}$ that will be compatible with the selection of a subset $\mathcal{Y}_i$ of quadratization items. 

Let $Z=\lg^2 t Y^2 \cdot 3^{\lg^2 t}$, and observe that $Z$ is larger than the profit of all encoding and quadratization items in total. For each $k \in \{0,\ldots,\lg t-1 \}$, we construct two \emph{index items}~$z^0_k$ and~$z^1_k$ corresponding to selecting either $i(k)=1$ or $i(k)=0$ in the binary representation $i(0),\ldots,i(\lg t-1)$ of~$i$. The weight and profit of these two items are defined by:
\begin{itemize}
\item $w(z^0_k) = p(z^0_k) = 2^k \cdot Z  + \sum^{\lg t -1}_{\ell=0} 3^{f(k,\ell)} \cdot Y$.
\item $w(z^1_k) = p(z^1_k) = 2^{k} \cdot Z  + 2^k \cdot 3B $.
\end{itemize} 
We use $\mathcal{Z}=\{z^0_k,z^1_k : 0 \leq k \leq \lg t-1\}$ to denote the set of all index items. 

Let $i \in \{0,\ldots,t -1\}$, and let $i(0),\ldots,i(t-1) \in \{0,1\}$ be binary values such that $i = \sum_k i(k) \cdot 2^k$. Observe that the set of index items $\mathcal{Z}_i$ defined by
$$
\mathcal{Z}_i = \{ z^{i(k)}_k : 0 \leq k \leq \lg t -1 \}
$$
naturally corresponds to index~$i$. In the lemma below, we show that due to our selection of the large value~$Z$, any set of items of sufficiently small weight and sufficiently large profit contains a subset of index items that correspond precisely to some index $i\in\{0,\ldots,t-1\}$. We let $w_Z(\cdot)$ denote the weight function $w_Z(x)=\lfloor w(x)/Z \rfloor$, and $p_Z(\cdot)$ denote the profit function $p_Z(x)=\lfloor p(x)/Z \rfloor$.
\begin{lemma}
\label{lem:index}%
Let $\mathcal{S} \subseteq \mathcal{X} \cup \mathcal{Y} \cup \mathcal{Z}$ be a set of items with $w_Z(\mathcal{S}) \leq t-1$ and $p_Z (\mathcal{S}) \ge t-1$. Then there exists some $i \in \{0,\ldots,t -1\}$ for which $\mathcal{S} \cap \mathcal{Z} = \mathcal{Z}_i$.
\end{lemma}

\begin{proof}
Let $\mathcal{Z}^* = \mathcal{S} \cap \mathcal{Z}$ denote the subset of index items in $\mathcal{S}$. 
As $w_Z(\mathcal{Z}^*) = p_Z (\mathcal{Z}^*)$, we have~$w_Z(\mathcal{Z}^*)=t-1$. Then due to the construction of the weights of the index items, it holds that
$$
w_Z(\mathcal{Z}^*) \,\,=\,\, \sum_{k=0}^{\lg t} |\{z^0_k,z^1_k\} \cap \mathcal{Z}^*| \cdot 2^k \,\,=\,\, t-1\enspace.
$$
Due to \Cref{lem:algebraic}, this equality can only hold if $|\{z^0_k,z^1_k\} \cap \mathcal{Z}^*|=1$ for all $k \in \{0,\ldots,\lg t -1\}$. Define 
$$
i(k) = 
\begin{cases}
0 &: z^0_k \in \mathcal{Z}^*\\
1 &: z^1_k \in \mathcal{Z}^*
\end{cases}
$$
for each $k \in \{0,\ldots,\lg t -1\}$. Then $\mathcal{Z}^* = \mathcal{Z}_i$ for $i = \sum^{\lg t- 1}_{k=0} i(k) \cdot 2^k$. 
\end{proof}

Now let us address the terms that depend on~$Y$ in the profit and weight of the instance selection items. Define $T$ to be the constant $T:= \sum^{\lg^2 t -1}_{k=0} 3^k$. Now consider some set~$\mathcal{Z}_i$ of index items corresponding to index $i\in\{0,\ldots,t-1\}$. Let $w_Y(\cdot)$ denote the weight function~$w_Y(x)=\allowbreak{\lfloor (w(x)-Z \cdot w_Z(x))/Y \rfloor}$, and similarly define the profit function $p_Y(\cdot)$ as~${p_Y(x)=\lfloor (p(x)-Z\cdot p_Z(x))/Y \rfloor}$. Then one can observe that, by construction of the weights and profits above, we have that both~${w_Y(\mathcal{Z}_i)+w_Y(\mathcal{Y}_i)}$ and $p_Y(\mathcal{Z}_i)+p_Y(\mathcal{Y}_i)$ equal $T$. Moreover, any other set of items with $w_Y$-weight at most $T$ will have lesser profit. This ensures a compatible selection of the index items and the quadratization items, formally proven in the lemma below.
\begin{lemma}
\label{lem:compatibility}%
Let $\mathcal{S} \subseteq \mathcal{X} \cup \mathcal{Y} \cup \mathcal{Z}$ be a set of items with weight $w_Z(\mathcal{S}) \leq (t-1) $, $w_Y(\mathcal{S}) \leq T$, $p_Z (\mathcal{S}) \ge (t-1)$, $p_Y (\mathcal{S}) \ge T$, and there is no set~$\mathcal{S}' \subseteq \mathcal{X} \cup \mathcal{Y} \cup \mathcal{Z}$ with $w(\mathcal{S}') \le w (\mathcal{S})$ and $p(\mathcal{S}') > p (\mathcal{S})$. Then $\mathcal{S} \cap \mathcal{Y} = \mathcal{Y}_i$ and $\mathcal{S} \cap \mathcal{Z} = \mathcal{Z}_i$ for some $i \in \{0,\ldots,t-1\}$.
\end{lemma}

\begin{proof}
Let $\mathcal{Z}^* = \mathcal{S} \cap \mathcal{Z}$ and $\mathcal{Y}^* = \mathcal{S} \cap \mathcal{Y}$. As $w_Z (\mathcal{S}) \le (t-1)$ and $p_Z (\mathcal{S}) \ge (t-1)$, by \Cref{lem:index} we have that $\mathcal{Z}^* = \mathcal{Z}_i$ for some $i \in \{0,\ldots,t-1\}$. Thus, to complete the proof, we focus on showing that~$\mathcal{Y}^* = \mathcal{Y}_i$. Let $i(0),\ldots,i(\lg t-1) \in \{0,1\}$ be such that $i = \sum_k i(k) \cdot 2^k$.

Observe that by the construction of the weights and profits of the index items, we have $w_Y(\mathcal{Z}_i)=p_Y(\mathcal{Z}_i)=\sum_{\alpha(k)=0} \sum_\ell 3^{f(k,\ell)} $. From this, one can see that both $w_Y(\mathcal{Z}_i)+w_Y(\mathcal{Y}^*)$ and $w_Y(\mathcal{Z}_i)+p_Y(\mathcal{Y}^*)$ equal
\begin{equation*}
\begin{split}
\Big( \sum_{\substack{i(k) = 0,\\ 0 \leq \ell \leq \lg t -1}} \!\!\! 3^{f(k,\ell)} \Big)  &+ \Big( \sum_{y^{1,0}_{k,\ell} \in \mathcal{Y}_i} \!\!\! 3^{f(k,\ell)}  + \!\!\! \sum_{y^{0,1}_{k,\ell} \in \mathcal{Y}_i} \!\!\! 3^{f(\ell,k)}  + \!\!\!\sum_{y^{1,1}_{k,\ell} \in \mathcal{Y}_i} \!\!\! (3^{f(k,\ell)}+3^{f(\ell,k)}) +
\!\!\!\sum_{y^{1,1}_{k,k} \in \mathcal{Y}_i} \!\!\! 3^{f(k,k)} \Big)  =\\ 
\Big(\sum_{\substack{i(k) = 0,\\ 0 \leq \ell \leq \lg t -1}} \!\!\! 3^{f(k,\ell)} \Big)  &+ \Big( \!\!\! \sum_{\substack{i(k) = 1,\\ i(\ell) = 0,\\ k< \ell}} \!\!\! 3^{f(k,\ell)} + \!\!\!\sum_{\substack{i(k) = 1,\\ i(\ell) = 0,\\ k> \ell}} \!\!\! 3^{f(k,\ell)}  + \!\!\!\sum_{\substack{i(k) = 1,\\ i(\ell) = 1,\\ k \neq \ell}} \!\!\! 3^{f(k,\ell)} + \!\!\! \sum_{i(k) = 1} \!\!\! 3^{f(k,k)} \Big)\enspace.\\
\end{split}
\end{equation*}
Note that as $w_Y(\mathcal{Z}_i) + w_Y(\mathcal{Y}^*) =w_Y(\mathcal{S}) \leq T$, the sum above is bounded from above by~$T$. Moreover, by \Cref{lem:algebraic}, the only way this reaches the bound with equality is if we have 
$$
w_Y(\mathcal{Y}^*)  \,\,=  \sum_{\substack{i(k)=1,\\ i(\ell) = 0,\\ k< \ell}} \!\!\! 3^{f(k,\ell)} + \!\!\!\sum_{\substack{i(k) = 1,\\ i(\ell) = 0,\\ k> \ell}} \!\!\! 3^{f(k,\ell)}  + \!\!\!\sum_{\substack{i(k) = 1,\\ i(\ell) = 1,\\ k \neq \ell}} \!\!\! 3^{f(k,\ell)} + \!\!\! \sum_{i(k) = 1} \!\!\! 3^{f(k,k)} \,\,=\,\, \sum_{\substack{i(k) = 0,\\ 0 \leq \ell \leq \lg t -1}} \!\!\! 3^{f(k,\ell)}
$$
which will then give us 
$$
w_Y(\mathcal{Z}_i) + w_Y(\mathcal{Y}^*)=\sum_{\substack{i(k) = 0,\\ 0 \leq \ell \leq \lg t -1}} \!\!\! 3^{f(k,\ell)}   + \sum_{\substack{i(k) = 1,\\ 0 \leq \ell \leq \lg t -1}} \!\!\! 3^{f(k,\ell)}  = \sum^{\lg^2 t -1}_{k=0} 3^k  = T\enspace.
$$
(Here, the penultimate equality follows because $f(\cdot,\cdot)$ is bijective.)

Note that by construction and \Cref{lem:algebraic}, the only sets of quadratization items $\mathcal{Y}^*$ with $w_Y(\mathcal{Y}^*)=p_Y(\mathcal{Y}^*)=\sum_{\alpha(k)=1} \sum_\ell 3^{f(k,\ell)}$ are either $\mathcal{Y}_i$, or any set of quadratization items obtained from $\mathcal{Y}_i$ by replacing some~$y^{1,1}_{k,\ell} \in \mathcal{Y}_i$ with $y^{1,0}_{k,\ell}$ and $y^{0,1}_{k,\ell}$.
If $\mathcal{S}$ contained $y^{1, 0}_{k, \ell}$ and $y^{0,1}_{\ell,k}$ for some $0 \le k < \ell \le \lg t-1$, then $\mathcal{S}' := (\mathcal{S} \setminus \{y^{1,0}_{k, \ell}, y^{0,1}_{k, \ell}\}) \cup \{y^{1,1}_{k, \ell}\}$ satisfies $w(\mathcal{S}') = w (\mathcal{S})$ and $p(\mathcal{S}') > p(\mathcal{S})$, a contradiction to the definition of~$\mathcal{S}$.
Thus,
we conclude that $\mathcal{Y}^*=\mathcal{Y}_i$ and the lemma is proven. 
\end{proof}

\subsection{Correctness}

Our entire \knapsack\ instance consists of all items $\mathcal{X} \cup \mathcal{Y} \cup \mathcal{Z}$. An overview of the weight and profit of each item can be found in \Cref{tab:wsharp}. We set the weight $W$ of the \knapsack\ instance to 
$$
W := (t-1) \cdot Z  \,\, + \,\, T \cdot Y  \,\, + \,\, (3t-2)\cdot B
$$
and the desired profit $P$ to
\begin{equation*}
\begin{split}
P &:=  W \,\ + \,\,  \binom{t}{2}  \cdot 9nB\\ 
&= (t-1) \cdot Z  \,\, + \,\, T \cdot Y  \,\, + \,\, (3t-2)\cdot B \,\, + \,\, \binom{t}{2}  \cdot 9nB\enspace.      
\end{split}
\end{equation*}
In the next two lemmas below we prove the correctness of our constructed composition. 

\begin{table}
\begin{center}
\begin{tabular}{c c c c}
\toprule
Item & Weight & Profit & Index Range \\

\midrule
\rule{0pt}{1.1em}
$z^1_k$ & $2^{k} \cdot Z + 2^k \cdot 3B$ & $2^k \cdot Z + 2^k \cdot 3B $& $0 \leq k \leq \lg t -1$\\

\midrule
\rule{0pt}{1.1em}
$z^0_k $ & $2^k \cdot Z + \sum_\ell 3^{f(k,\ell)} \cdot Y$ & $2^k \cdot Z + \sum_\ell 3^{f(k,\ell)} \cdot Y$ & $0 \leq k \leq \lg t -1$ \\

\midrule
\rule{0pt}{1.1em}
$y^{1,0}_{k, \ell} $ & $3^{f(k,\ell)} \cdot Y$ & $3^{f(k,\ell)} \cdot Y$ & $0 \leq k < \ell \leq \lg t -1$\\

\midrule
\rule{0pt}{1.1em}
$y^{0,1}_{k, \ell} $ & $3^{f(\ell,k)} \cdot Y$ & $3^{f(\ell,k)} \cdot Y$ & $0 \leq k < \ell \leq \lg t -1$\\

\midrule
\rule{0pt}{1.1em}
$y^{1,1}_{k, \ell} $ & $(3^{f(k,\ell)} + 3^{f(\ell,k)})\cdot Y$ & $(3^{f(k,\ell)} + 3^{f(\ell,k)})\cdot Y + 2^{k+\ell} \cdot 9nB$ & $0 \leq k < \ell \leq \lg t -1$ \\

\midrule
\rule{0pt}{1.1em}
$y^{1,1}_{k, k} $ & $3^{f(k,k)} \cdot Y$ & $3^{f(k,k)} \cdot Y + 2^{k+k} \cdot 4.5nB + 2^k \cdot 1.5nB$ & $0 \leq k \leq \lg t -1$ \\

\midrule
\rule{0pt}{1.1em}
\multirow{2}*{$x^i_j$} & \multirow{2}*{$X + a^i_j$} & \multirow{2}*{$X+a^i_j + i \cdot 3B$} & $0 \leq i  \leq t-1$,\\
 &  &  & $1 \leq j \leq 3n$ \\
\bottomrule
\end{tabular}
\caption{The weights and profits of the items used in the proof of \Cref{thm:main} for parameter~$w_{\#}$. The three large constants used in the proof are $X = 3tn \cdot B_n$, $Y=t^2 \cdot 3nB$, and $Z=\lg^2 t Y^2 \cdot 3^{\lg^2 t}$.}
\label{tab:wsharp}%
\end{center}
\end{table}

\begin{lemma}
\label{lem:ForwardDirection}%
If $\mathcal A_i$ is a ``yes"-instance of \rss\ for some $i \in \{0,\ldots,t-1\}$ then there exists a subset of items $\mathcal{S} \subseteq \mathcal{X} \cup \mathcal{Y} \cup \mathcal{Z}$ with $w(\mathcal{S}) \leq W$ and $p(\mathcal{S}) \geq P$.
\end{lemma}

\begin{proof}
Suppose $\mathcal A_i$ is a ``yes"-instance of \rss\ for some $i \in \{0,\ldots,t-1\}$, and let $i(0),\ldots,i(\lg t- 1) \in \{0,1\}$ be such that $i = \sum_k i(k) \cdot 2^k$. Then $w(\mathcal{X}(\mathcal A_i))=(3t-3i-2) \cdot B$ as discussed in \Cref{subsec:encoding}. Furthermore, $w(\mathcal{Y}_i)=w_Y(\mathcal{Y}_i) \cdot Y= (T-w_Y(\mathcal{Z}_i))\cdot Y$ as is shown in the proof of \Cref{lem:compatibility}. Finally, we have $w_Z(\mathcal{Z}_i)=(t-1) \cdot Z$, and 
$$
w(\mathcal{Z}_i) \,\,=\,\, w_Z(\mathcal{Z}_i) + w_Y(\mathcal{Z}_i) + \!\!\!\sum_{i(k)=1} 2^k \cdot 3B \,\,=\,\, (t-1) \cdot Z + w_Y(\mathcal{Z}_i) \cdot Y + i \cdot 3B\enspace.
$$  
So altogether we have 
\begin{equation*}
\begin{split}
w(\mathcal{X}(\mathcal A_i) \cup \mathcal{Y}_i \cup \mathcal{Z}_i) \,\,&=\,\, (3t-3i-2) \cdot B \,\,+\,\, (T-w_Y(\mathcal{Z}_i))\cdot Y\\
&+(t-1) \cdot Z \,+\, w_Y(\mathcal{Z}_i)\cdot Y \,+\, i \cdot 3B \\
& = (t-1) \cdot Z \,+\, T\cdot Y \,+\, (3t-2) \cdot B \,\, = W\enspace. 
\end{split}
\end{equation*}

Let us next calculate the profit of $\mathcal{X}(\mathcal A_i) \cup \mathcal{Y}_i \cup \mathcal{Z}_i$. Recall that $p(\mathcal{X}(\mathcal A_i))=w(\mathcal{X}(\mathcal A_i))+(\binom{t}{2}-\binom{i+1}{2} + \frac{i}{3}) \cdot 9nB$ as discussed in \Cref{subsec:encoding}. By \Cref{lem:quadratization} we have $p(\mathcal{Y}_i)=w(\mathcal{Y}_i)+(\binom{i+1}{2} - \frac{i}{3}) \cdot 9nB$, and by construction we have $p(\mathcal{Z}_i)=w(\mathcal{Z}_i)$. Thus, altogether we have 
\begin{equation*}
\begin{split}
p(\mathcal{X}(\mathcal A_i) \cup \mathcal{Y}_i \cup \mathcal{Z}_i) \,\,&=\,\, w(\mathcal{X}(\mathcal A_i))+(\binom{t}{2}-\binom{i+1}{2} + \frac{i}{3}) \cdot 9nB\\ 
&+ w(\mathcal{Y}_i)+(\binom{i+1}{2} - \frac{i}{3}) \cdot 9nB \,\,+\,\,  w(\mathcal{Z}_i)\\
&= w (\mathcal{X}(\mathcal A_i) \cup \mathcal{Y}_i \cup \mathcal{Z}_i) + \binom{t}{2} \cdot 9nB = P\enspace.
\end{split}
\end{equation*} 
Thus the set of items $\mathcal{S}=\mathcal{X}(\mathcal A_i) \cup \mathcal{Y}_i \cup \mathcal{Z}_i$ is a solution for our \knapsack\ instance, and the lemma is proven. 
\end{proof}

\begin{lemma}
\label{lem:BackwardDirection}%
If there exists a subset of items $\mathcal{S} \subseteq \mathcal{X} \cup \mathcal{Y} \cup \mathcal{Z}$ with $w(\mathcal{S}) \leq W$ and $p(\mathcal{S}) \geq P$ then there is some $i \in \{0,\ldots,t-1\}$ for which $\mathcal A_i$ is a ``yes"-instance of \rss.
\end{lemma}

\begin{proof}
Let $\mathcal{S}$ be a solution with $w(\mathcal{S}) \leq W$ and $p(\mathcal{S}) \geq P$. Let $\mathcal{X}^* = \mathcal{S} \cap \mathcal{X}$, $\mathcal{Y}^* = \mathcal{S} \cap \mathcal{Y}$, and~${\mathcal{Z}^* = \mathcal{S} \cap \mathcal{Z}}$. Then as $w(\mathcal{S}) \leq W < t \cdot Z$ and $p(\mathcal{S}) \le P < t \cdot Z$, we have by \Cref{lem:index} that $\mathcal{Z}^*=\mathcal{Z}_i$ for some~${i \in \{0,\ldots,t-1\}}$. Furthermore, as we can assume $\mathcal{S}$ is of maximal profit (\ie there is no set~$\mathcal{S}' $ with $w(\mathcal{S}') \le w(\mathcal{S})$ and $p(\mathcal{S'}) > p (\mathcal{S})$) we have $\mathcal{Y}^*=\mathcal{Y}_i$
by \Cref{lem:compatibility}. As shown in the proof of \Cref{lem:ForwardDirection}, we have $w(\mathcal{Z}_i)+w(\mathcal{Y}_i)= (t-1) \cdot Z + T \cdot Y +i \cdot 3B$. Thus, 
$$
w(\mathcal{X}^*) \leq W - w(\mathcal{Z}_i)  - w(\mathcal{Y}_i) = (3t-3i-2) \cdot B\enspace. 
$$
Moreover, by \Cref{lem:quadratization} we have $p(\mathcal{Y}_i)=w(\mathcal{Y}_i)+(\binom{i+1}{2} - \frac{i}{3}) \cdot 9nB$, and by construction we have $p(\mathcal{Z}_i)=w(\mathcal{Z}_i)$. Thus, 
\begin{equation*}
\begin{split}
p(\mathcal{X}^*) &\geq P - p(\mathcal{Z}_i)  - p(\mathcal{Y}_i) \\
& = P -  w(\mathcal{Z}_i)  - w(\mathcal{Y}_i) - \Bigl(\binom{i+1}{2} - \frac{i}{3}\Bigr) \cdot 9nB\\
& = P -  (t-1) \cdot Z  - T \cdot Y - i \cdot 3B - \Bigl(\binom{i+1}{2} - \frac{i}{3}\Bigr) \cdot 9nB\\
& = (3t-3i-2) \cdot B +  \Bigr(\binom{t}{2} - \binom{i+1}{2} + \frac{i}{3}\Bigr) \cdot 9nB\enspace. 
\end{split}    
\end{equation*}
It therefore follows by \Cref{lem:encoding} that instance~$\mathcal A_i$ is indeed a ``yes"-instance of \rss, and the lemma follows.  
\end{proof}

\begin{proof}[Proof of \Cref{thm:main}]
We presented above an algorithm that composes any~$t$ instances $\mathcal A_0,\ldots,\mathcal A_{t-1}$ of \rss\ into a single instance of \knapsack\ in polynomial-time. By \Cref{lem:ForwardDirection,lem:BackwardDirection}, the constructed \knapsack\ instance is a ``yes"-instance if and only if~$\mathcal A_i$ is ``yes"-instance of \rss\ for some $i \in \{0,\ldots,t-1\}$. Observe that total number of different weights in our constructed \knapsack\ instance is 
$$
w_{\#} \leq |\widetilde{\mathcal{A}}_n| + |\mathcal{Y}| + |\mathcal{Z}| = O(n^3+\lg^2 t)\enspace. 
$$
Thus our algorithm fulfills all requirements of a composition algorithm, as given in \Cref{def:composition}. The proof for $w_{\#}$ then follows by a direct application of \Cref{thm:NoPolyKernel}. 
The statement for $p_{\#}$ follows by applying the reduction from Polak \etal~\cite[Chapter~4]{PolakEtAl21} which reduces an instance with $w_{\#} = k$ different item weights to an instance with~$p_{\#}= k$ different item profits.
\end{proof}


\section{Polynomial Kernel for Parameter $w_{\#} + p_{\#}$}

In this section we present a polynomial kernel for \knapsack\ parameterized by~$w_{\#} + p_{\#}$, thereby proving \Cref{thm:secondary}. Our kernel is a direct generalization of the polynomial kernel for \subsum\ parameterized by $a_{\#}$ of Etscheid \etal~\cite{EtscheidKMR17}.

The presented kernel utilizes two classical results: First, we use the fact that integer programming is fixed-parameter tractable with respect to the number of variables (this was first shown by Lenstra~\cite{Lenstra83}, and the currently best known running time with respect to the number of variable is due to Reis and Rothvoss~\cite{ReisR23}):
\begin{theorem}[\cite{ReisR23}]
\label{thm:ILP}%
\textsc{Integer Linear Programming} with input size $s$ and $n$ variables can be solved in $ 2^{O(n\lg\lg n)} \cdot s^{O(1)}$ time.
\end{theorem}
\noindent Second, we use the following theorem by Frank and Tardos~\cite{FrankT87}:
\begin{theorem}[\cite{FrankT87}]
\label{thm:loosing-weight}%
    There is an algorithm that, given a vector $w \in \mathbb{Q}^r$ and a natural number~$N$, computes in polynomial time a vector $\overline w \in \mathbb{Q}^r$ with $\|w\|_\infty \le 2^{4r^3} \cdot N^{r^2+2r}$ and $\sign (w \cdot b) = \sign (\overline w \cdot b)$ for every $b \in \mathbb{Z}^r$ with $\|b\|_1 \le N$.
\end{theorem}

\begin{proof}[Proof of \Cref{thm:secondary}]
Let $w_1, \ldots, w_{w_{\#}}$ be the different weights and $p_1, \dots, p_{p_\#}$ be the different profits in a given \knapsack\ instance with $n$ items. We denote by $n_{i,j}$ for $i \in \{1,\ldots,w_{\#}\}$ and ${j \in \{1,\ldots,p_\#\}}$ the number of items with weight~$w_i$ and profit $p_j$. First note that the following is an ILP formulation of \knapsack\ with $w_{\#} \cdot p_{\#}$ many variables $x_{i,j}$, one for each $i \in \{1,\ldots,w_{\#}\}$ and~${j \in \{1,\ldots,p_\#\}}$, and two inequalities:
\begin{ILP}
\begin{aligned}\label{ILP}
\sum_{i=1}^{w_\#} \sum_{j=1}^{p_\#} x_{i,j} \cdot w_i &\le W \\
\sum_{i=1}^{w_\#} \sum_{j=1}^{p_\#} x_{i,j} \cdot p_j &\ge P \\
x_{i,j} &\in\{0,1,\ldots, n_{i,j}\}.
\end{aligned}
\end{ILP}
    
By \Cref{thm:ILP}, if $w_{\#} \cdot p_{\#} \cdot \lg (w_{\#} \cdot p_{\#}) \le \lg n$ (recall that $n$ denotes the total number of items of the \knapsack\ instance), then the instance can be solved in polynomial time using \Cref{ILP}. Thus, devising a polynomial kernel in this case is trivial. So assume that $w_{\#} \cdot p_{\#} \cdot \lg (w_{\#} \cdot p_{\#}) > \lg n$. To reduce the encoding length of \Cref{ILP}, we apply \Cref{thm:loosing-weight} to the first two inequalities of \Cref{ILP} as follows: Let $w^*$ be $w_\# \cdot p_\#$-dimensional vector whose $(p_\# \cdot (i -1) + j)$'th component equals~$w_i$, for each $i \in \{1,\ldots,w_{\#}\}$ and $j \in \{1,\ldots,p_\#\}$. We apply \Cref{thm:loosing-weight} to the $(w_\# \cdot p_\#+1)$-dimensional vector $w:= (w^*,-W)$ and $N := n+1$, resulting in a vector~$(\overline w^*, -\overline W)$. Similarly, let~$p^*$ be $w_\# \cdot p_\#$-dimensional vector whose $(p_\# \cdot (i - 1) + j)$'th component equals $p_j$, for $i \in \{1,\ldots,w_{\#}\}$ and $j \in \{1,\ldots,p_\#\}$. We apply \Cref{thm:loosing-weight} to the vector $w:= (-p^*, P)$ and $N := n+1$, resulting in a vector~$(-\overline p^*, \overline P)$. By \Cref{thm:loosing-weight}, \Cref{ILP} is equivalent to the following \Cref{ILP:reduced}:
\begin{ILP}
\begin{aligned}\label{ILP:reduced}
\overline w^* \cdot x &\le \overline W\\
\overline p^* \cdot x &\ge \overline P\\
x_{i,j} & \in \{0,1,\ldots,n_{i,j}\}
\end{aligned}
\end{ILP}By \Cref{thm:loosing-weight}, each number from \Cref{ILP:reduced} has encoding length $O(r^3 +r^2 \lg n)$ for $r = w_\# \cdot p_\#$.
Since $\lg n \le r \cdot \lg r$, it follows that the size of \Cref{ILP:reduced} is $O(r^4 \cdot \lg r)$. As \knapsack\ is \NP-complete, and Integer Linear Programming is in \NP, we can reduce \Cref{ILP:reduced} in polynomial time to an instance of \knapsack. The resulting \knapsack\ instance  is equivalent to the original instance, and has size polynomial in~$r = w_\# \cdot p_\#$.
\end{proof}
Using the same trick as Etscheid \etal~\cite{EtscheidKMR17} for \textsc{Subset Sum}, one can also make the reduction from \Cref{ILP:reduced} to \knapsack\ explicit, resulting in a kernel of size $\widetilde O (r^5)$.

\bibliography{knapsack}

\end{document}